\documentclass[11pt, a4paper]{article}

\usepackage{a4wide,url}
\usepackage{graphicx,amssymb,amsmath,color,amsthm}

\newtheorem{lemma}{Lemma}
\newtheorem{theorem}[lemma]{Theorem}

\begin{document}
\pagestyle{plain}

\title{Notes on Convex Transversals}
 
\author{
Lena Schlipf\thanks{Institute of Computer Science, Freie Universit\"at Berlin, Germany. {\tt schlipf@mi.fu-berlin.de}. This research was supported by the DFG within the Priority Programme 1307 Algorithm
Engineering. }}

\date{}
\maketitle

\begin{abstract}
In this paper, we prove the problem of stabbing a set of disjoint bends by a convex stabber to be NP-hard. We also consider the optimization version of the convex stabber problem and prove this problem to be APX-hard for sets of line segments.   
\end{abstract}

\section{Introduction}

Consider a finite set of geometric objects in the plane. We call this set stabbable if there exists a convex polygon whose boundary intersects every object. The boundary is then called \emph{convex stabber} or \emph{convex transversal}. 

The problem of finding a convex stabber was originally proposed by Tamir in 1987 \cite{tamir-87}. 
Arkin et al.~\cite{DBLP:conf/wads/ArkinDKMPSY11} proved this problem to be NP-hard when the geometric objects are line segments. They also proved that the problem remains NP-hard when the objects are similar copies of a given convex polygon.
In 3 dimensions, they showed that the problem is NP-hard for a set of balls.

This paper, in fact, can be considered as a continuation of the paper of Arkin et al. 
We will show that the problem of finding a convex stabber for a set of disjoint simple polygons is NP-hard. Actually, we even show that it is already hard to stab a set of disjoint bends. 
Additionally, we study the optimization version: Given a finite set of geometric objects in the plane, compute the maximum number of objects that can be stabbed with the boundary of a convex polygon.
We prove this problem to be APX-hard when the objects are line segments or similar copies of a given convex polygon.

\paragraph{Notation.}
Two line segments that have a common endpoint are called a \emph{bend}. 
We say that a convex stabber stabs or \emph{traverses} the given objects in a specific way.

\section{Convex Stabbers for Disjoint Polygons}\label{section:disjoint}

We consider the following problem
\begin{quote} Given a set of disjoint bends in the plane, is there a convex stabber that intersects every bend of the set?
\end{quote}

We show that this problem is NP-hard. We reduce from planar, monotone 3SAT which was shown to be NP-hard by de~Berg and Khosravi~\cite{DBLP:conf/cocoon/BergK10}. 
A \emph{monotone} instance of 3SAT is an instance where each clause has either only positive or only negatives variables. In the following we call a clause that contains only positive variables a \emph{positive clause} and a clause that contains only negative variables a \emph{negative clause}. De~Berg and Khosravi  \cite{DBLP:conf/cocoon/BergK10} also pointed out that planar monotone 3SAT remains NP-hard when a monotone rectilinear representation is given. In a monotone rectilinear representation the variable and clause gadgets are represented as rectangles. All variable rectangles lie on a horizontal line.  The edges connecting the clause gadgets to the variable gadgets are vertical line segments and no two edges cross. All positive clauses lie above the variables and all negative clauses lie below the variables. See Figure~\ref{monotonerectilinear3SAT} for an example of a monotone rectilinear representation of a planar 3SAT instance. 
\begin{figure}[h]\centering\includegraphics[scale=.8]{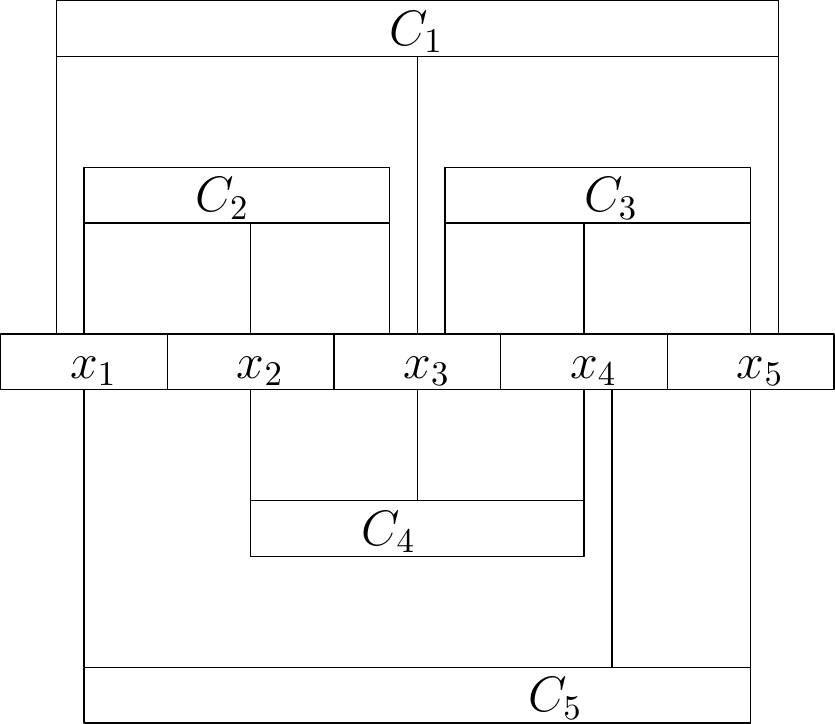}
\caption{A monotone rectilinear representation of the 3SAT instance \mbox{ $C=C_1\wedge C_2\wedge C_3\wedge C_4\wedge C_5$} where \mbox{$C_1=x_1\vee x_3\vee x_5$,} \mbox{$C_2=x_1\vee x_2\vee x_3$,} \mbox{ $C_3=x_3\vee x_4\vee x_5$,} \mbox{$C_4=\overline{x_2}\vee \overline{x_3}\vee \overline{x_4}$,} and \mbox{$C_5=\overline{x_1}\vee \overline{x_4}\vee \overline{x_5}$.}}\label{monotonerectilinear3SAT}\label{fig:3SAT}\end{figure}

Given a monotone rectilinear representation $\phi$ of a 3SAT instance, we construct a set of bends $\mathcal B$ such that there exists a convex stabber for $\mathcal B$ if and only if $\phi$ is satisfiable. 
 Let $m$ be the number of clauses and $n$ be the number of variables. Let the number of positive clauses and negative clauses be $m_1$ and $m_2$, respectively. 
(The variable and the clause gadgets are basically constructed in the same way as in~\cite{DBLP:conf/wads/ArkinDKMPSY11}.)

\paragraph*{Variable gadgets.} A variable gadget consists of a line segment and three points (degenerate bends), see Figure~\ref{fig:VariableGadget}. There are two ways to traverse these points and the segment depending on the order in which the middle point is traversed: one corresponds to setting the variable to True, the other to setting the variable to False.

\begin{figure}[h]\centering\includegraphics[scale=1]{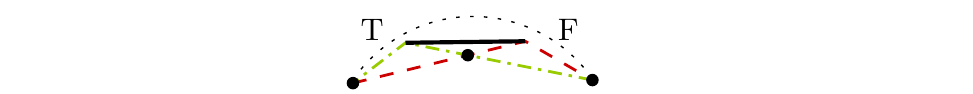}\caption{There are two ways to traverse the gadget: The dashed-dotted path corresponds to setting the variable True, the dashed path corresponds to setting the variable False. Clearly,  the dashed and the dashed-dotted segments and the dotted circular arc are not part of the construction.}\label{fig:VariableGadget}
\end{figure}

The variable gadgets are fitted into a circular arc by putting the non-middle points on the arc. The middle point and the line segment lie inside the circle. Each positive variable gadget is placed into an arc of $1/(16m_1)$ of a unit circle. Each negative variable gadgets is placed into an arc of $1/(16m_2)$ of a unit circle.

\paragraph*{Clause gadgets.}The clause gadgets consists of a line segment and two points (degenerate bends). 
Similarly to the variable gadgets the clause gadgets are fit into a circular arc. The points are put on the arc and the segment lies inside the arc.  The positive clause gadgets are fit into an arc of $1/(16m_1)$ of a unit circle. The negative variable gadgets are fit into an arc of $1/(16m_2)$ of a unit circle.

\paragraph*{Positive arc.} 
We place the gadgets representing a positive variable or a positive clause next to each other on an arc of  a unit circle. Since any positive variable and any positive clause gadget is fit into an arc of $1/(16m_1)$, they occupy an arc of one quarter of a unit circle.
The gadgets have to be placed in a specific order. Consider the monotone rectilinear representation. The clauses and their corresponding variables are connected via edges. We place a gadget for each edge, representing the variable that this edge connects with a clause. The edges have a specific order from left to right and the gadgets are sorted in the same way. Hence, a gadget is placed for every occurrence of a variable in a clause. The clause gadgets are placed to the right of the gadget representing their middle variable, see Figure~\ref{DisjointReduction3}.

\begin{figure}[h]
\centering\includegraphics[scale=.8]{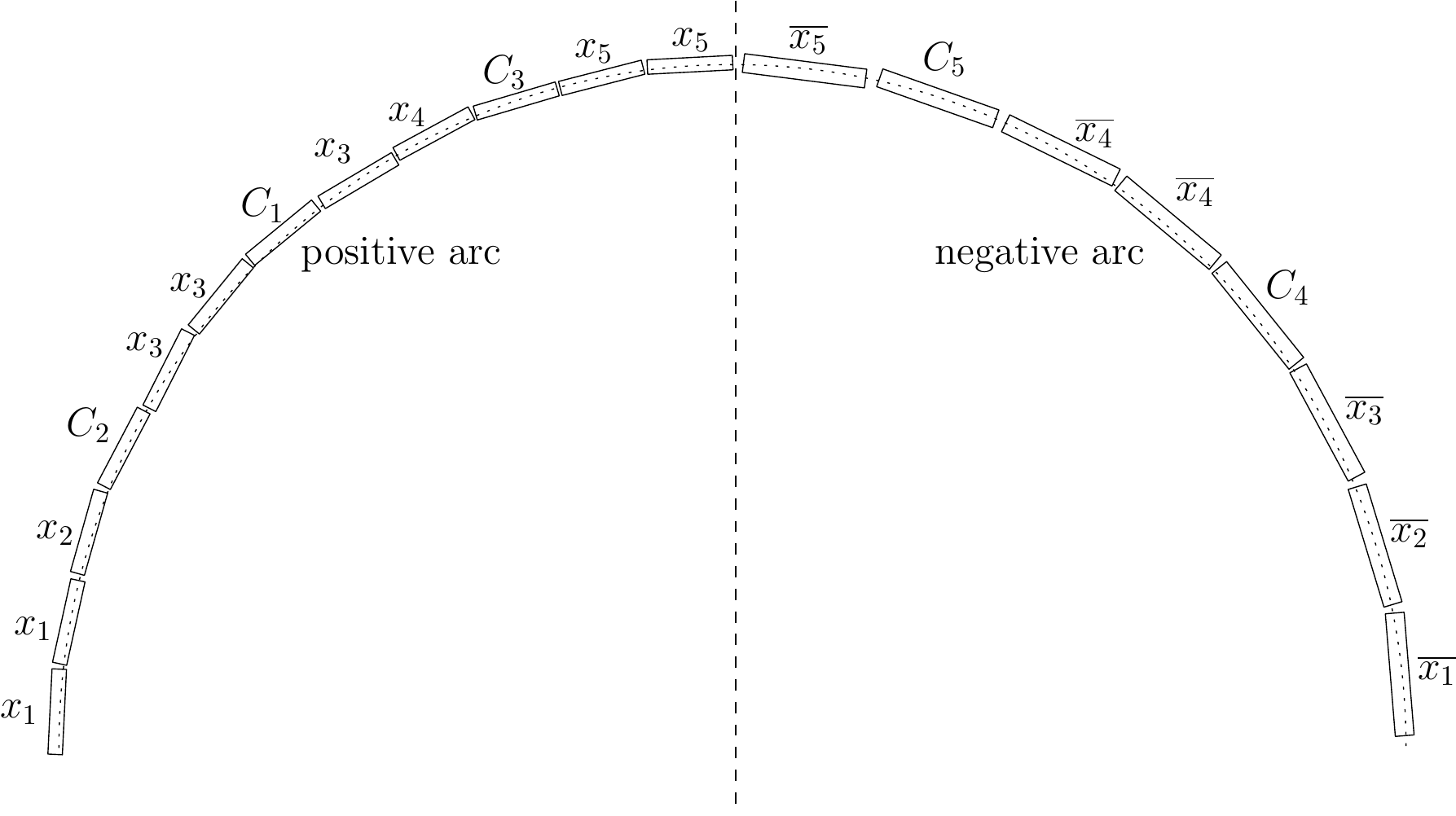}\caption{The placement of the gadgets for the instance in Fig.~\ref{fig:3SAT} is shown.}\label{DisjointReduction3}
\end{figure}

\paragraph*{Negative arc.}
We place the gadgets representing a negative variable or a negative clause next to each other on an arc of one quarter of a unit circle.
The gadgets are ordered in the same way as for the positive arc, but this time from right to left.
The negative and the positive arc are placed next to each other, see Figure~\ref{DisjointReduction3}.

\paragraph*{Variable connectors.} To ensure that a stabber has to traverse all gadgets representing the same variable in the same way, we place $3m$ variable connectors. All variable gadgets that represent the same variable are connected via segments in a circular manner.  
The segment touches the True path of one gadget and the False path of the next gadget (see Figure~\ref{DisjointVariableGadgets}). Since on both positive and negative arc all variable gadgets representing the same variable lie next to each other, these segments do not intersect anything else.
In total, we place one segment for each variable gadget.

\begin{figure}[h]\centering\includegraphics[scale=.55]{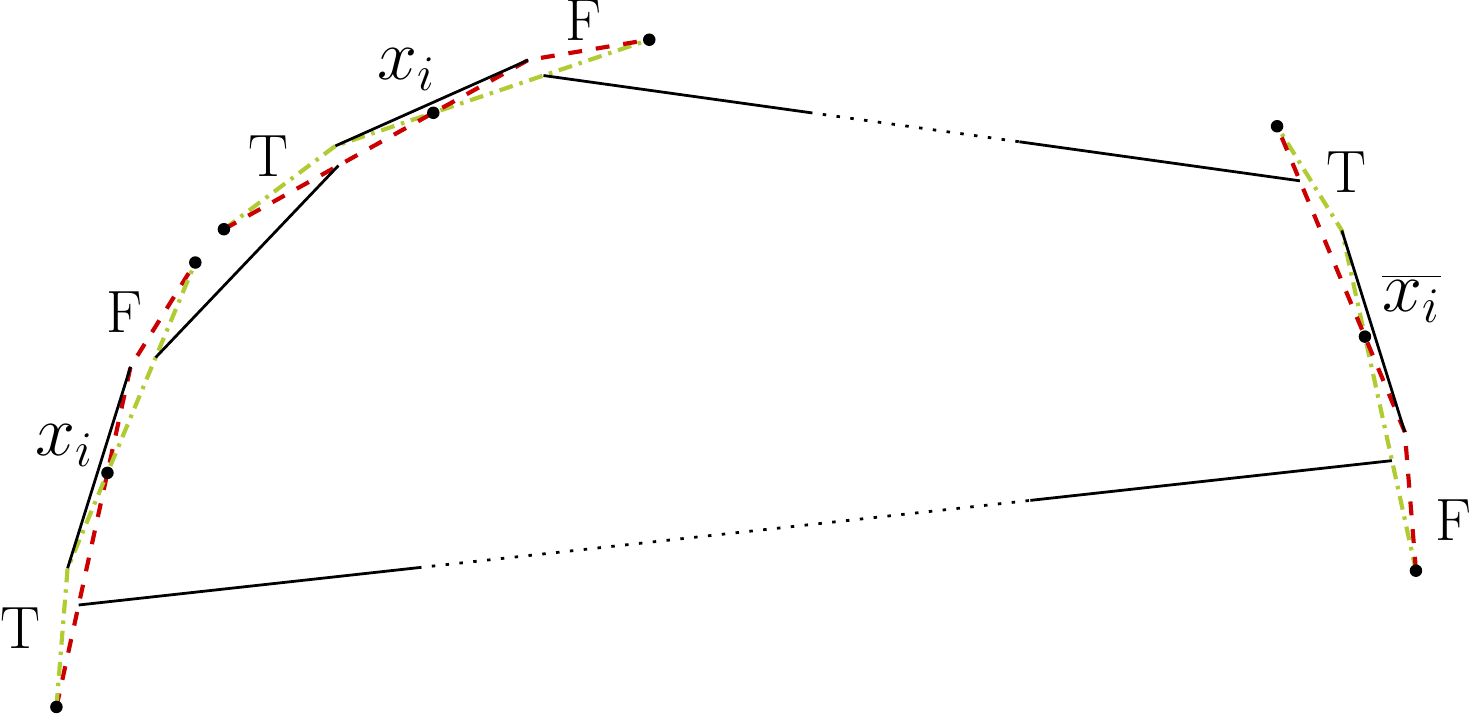}
\caption{The variable connectors are shown. These connectors ensure that each variable gadget that represents the same variable is traversed in the same way: either the stabber traverses the True path or the False path.}\label{DisjointVariableGadgets}\end{figure}

\paragraph*{Clause connectors.} We place $3m$ more bends in order to connect a clause gadget with its variables. Note that there is a variable gadget for each occurrence of a variable in a clause. 
The connectors either lie inside the circle or outside. We call them inner connectors and outer connectors, respectively. An inner connector can be a straight line segment whereas an outer connector has to be a bend.

\begin{figure}\centering\includegraphics[scale=.6]{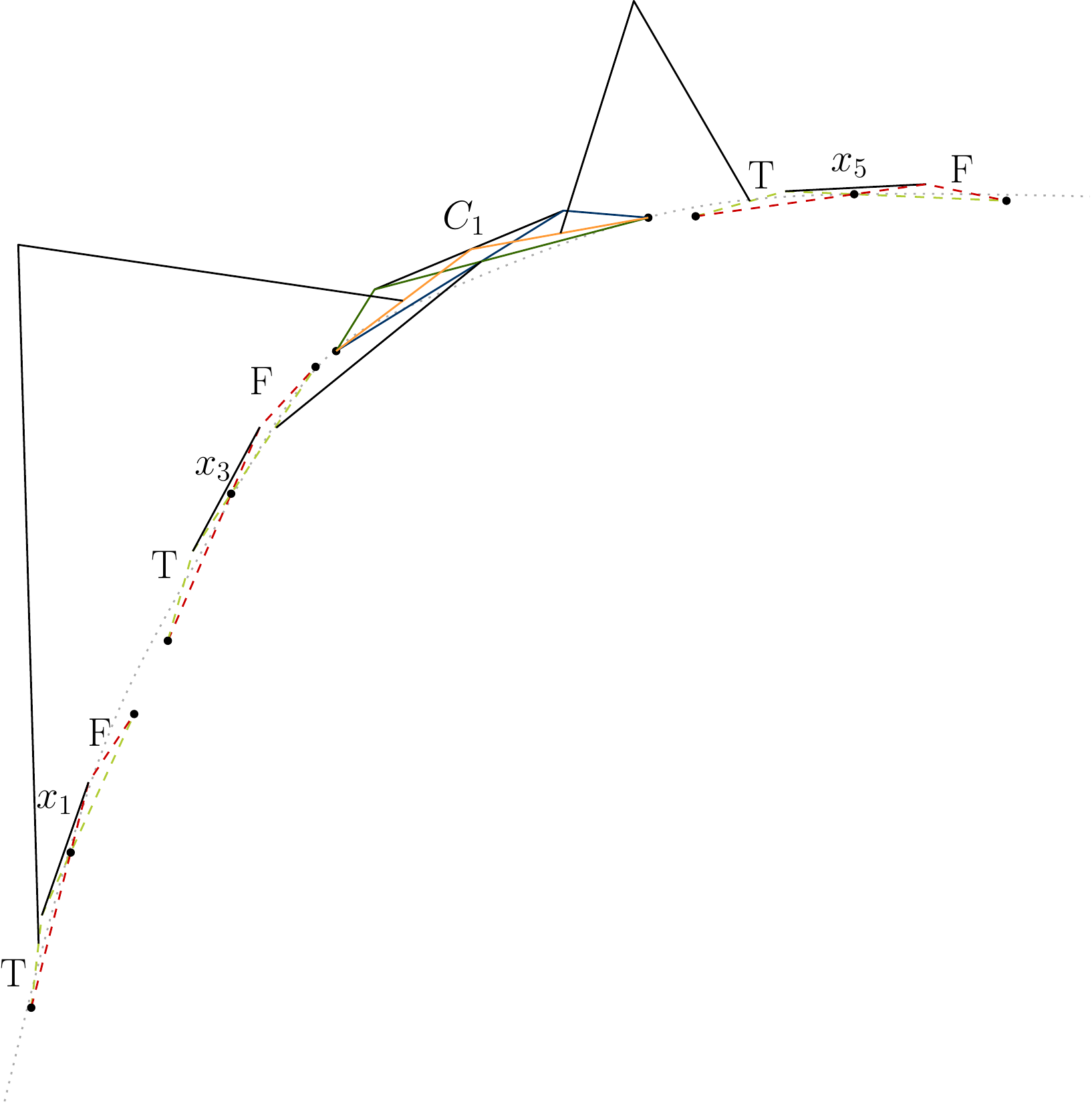}
\caption{The clause connectors are shown. There are three ways to traverse the clause gadget. Each path stabs two out of the three connectors. }\label{DisjointClauseGadgets}\end{figure}

The remaining parts of the construction are explained for positive clauses. Negative clauses can be handled similarly. 
Each clause gadget has two outer connectors and one inner connector. The inner connector connects the clause gadget to the gadget representing its middle variable. Note that this gadget is placed next to the clause gadget in the construction. The other two variables that occur in this clause are connected via outer connectors.
One endpoint of each connector lies within the variable gadget as follows: the segment touches the True path through the gadget and does not intersect the False subpath. In every clause gadget, the endpoints of these connectors look the same -- see Fig.~\ref{DisjointClauseGadgets}. It can easily be checked that a convex stabber can intersect any two of the three bends but never all three of them.
It follows immediately from the monotone rectilinear representation that our construction can be drawn crossing-free, hence the bends and segments are all pairwise disjoint.

\paragraph*{Correctness.}

\begin{figure}
\centering\includegraphics[scale=.8]{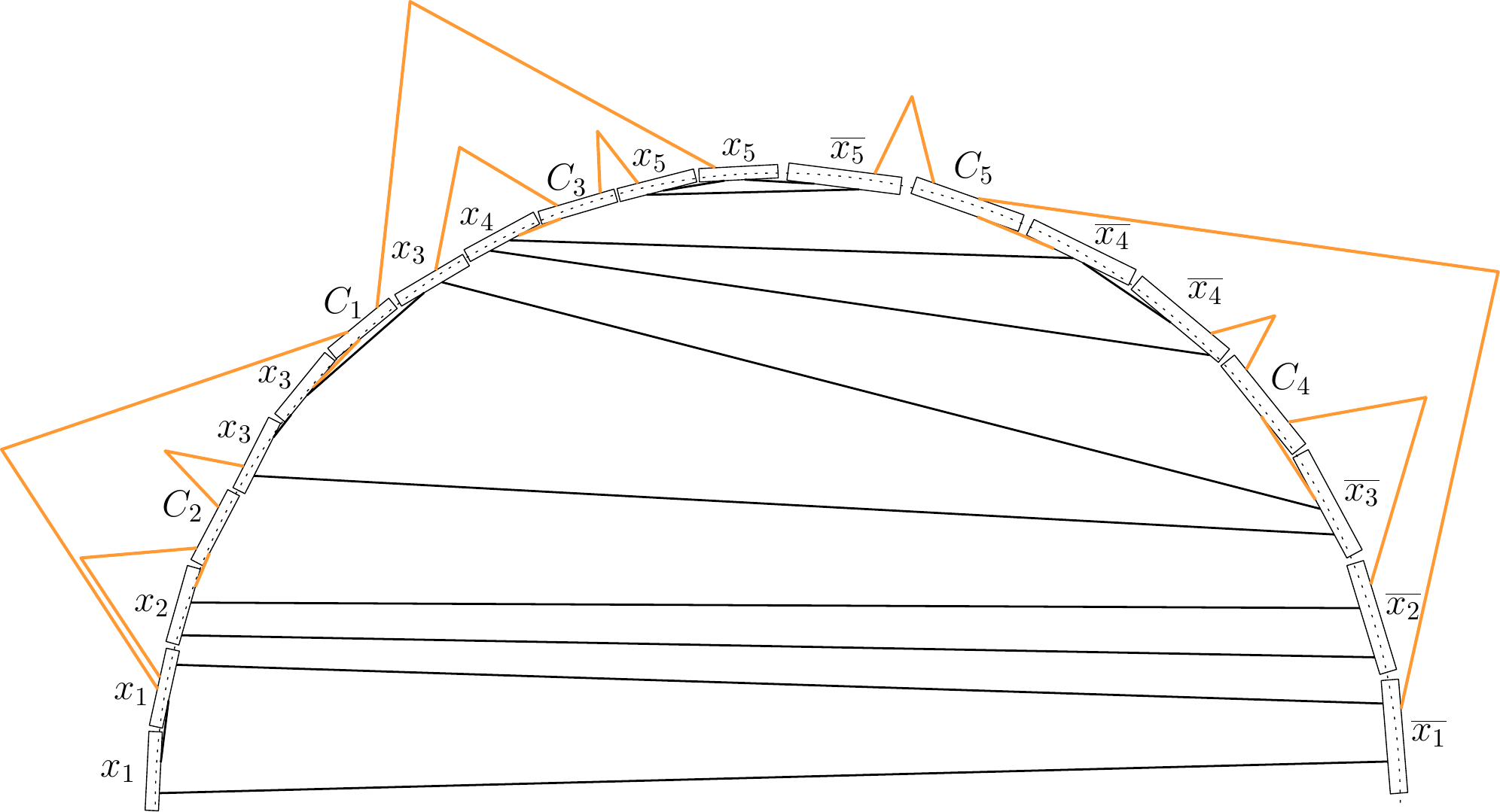}\caption{A sketch of the construction for the instance in Fig.~\ref{fig:3SAT} is shown. The clause connectors are marked in fat and orange.}\label{DisjointReduction}
\end{figure}
Assume there exists a satisfying assignment for the 3SAT formula.  The convex stabber traverses the variable gadgets according to this assignment. In each clause gadget the stabber can stab two connector bends. Since at least one variable is satisfied in each clause, the stabber can omit the connector  connecting the satisfied variable to the clause and stab the other two connectors. 
Hence, the stabber stabs all bends.

On the other hand, assume there is a convex stabber that stabs all bends. Set the variables True or False depending on how the  stabber traverse the variable gadget. This is a satisfying assignment for the 3SAT instance: The setting is consistent since the stabber has to omit at least one connector in each clause gadget. And hence these omitted bends have to be stabbed in the variable gadgets. And there the stabber can either take the True path or the False path, but not both.

\begin{theorem}Let $\mathcal B$ a set of disjoint bends in the plane. It is NP-hard to decide whether there exists a convex stabber for $\mathcal B$.
\end{theorem}

\section{APX-Hardness}\label{section:apxhardness}

In this section we consider the optimization version of the convex stabber problem which we call the \emph{maximum convex stabber problem}.
\begin{quote}
Given a set of geometric objects in the plane, compute a convex stabber that stabs the maximum number of these objects.
\end{quote}

We prove this problem to be APX-hard if the geometric objects are line segments. 
This result can then be generalized to the problem where the input is a set of similar copies of a given convex polygon.

We use a reduction from MAX-E3SAT(5). MAX-E3SAT(5) is a special version of MAX-3SAT where each clause contains exactly 3 literals and every variable occurs in exactly $5$ clauses and a variable does not  appear in a clause more than once. It is known to be NP-hard to approximate MAX-E3SAT(5) within a factor of $(1+\epsilon)$ for any $\epsilon >0$ \cite{DBLP:journals/jacm/Feige98}.
We first start by reducing MAX-E3SAT(5) to the decision version of the problem (Given a set of line segments and an integer $s$, does there exists a convex stabber that stabs $s$ segments?). 
The reduction is basically the same as the reduction for the problem of finding a convex stabber for line segments in~\cite{DBLP:conf/wads/ArkinDKMPSY11}. 

We show how to build a set of line segments $\mathcal L$, for a 3SAT formula $\phi$, such that there exists a convex stabber for $\mathcal L$ if and only if $\phi$ is satisfiable.
We use $n,m$ to denote the number of the variables and clauses. Here $5n=3m$.

\paragraph*{Variable gadgets.}
For each variable we have a gadget that consists of $6$ line segments and $18$ points (line segments of length 0). 
The line segments are stacked on top of each other. The $18$ points are partitioned into $3$ sets of equal size. The points of each set are stacked on top of each other. The stack of line segments and the $3$ stacks of points are arranged in the same way as in the NP-hardness proof for disjoint bends.
There are two ways to traverse the gadget (and to stab all segments in the gadget). One way corresponds to setting the variable to True, the other to setting the variable to False. Notice that any other way to traverse the gadget cannot stab all segments/points of this gadget. Thus, a stabber traversing any other way stabs at least $6$ segments less than a stabber traversing the True or the False path.

We fit the variable gadget into a circular arc by putting the two non-middle stack of points on the arc. The middle stack of points and the stack of segments lie inside the circle. Each variable gadget is fit into an arc of $1/(4n)$ of a unit circle.

 The variable gadgets are placed next to each other on an arc of one quarter of a unit circle. We call this arc the \emph{variable arc}. 

\paragraph*{Clause gadgets.}
Each clause gadget consists of $4$ segments and $8$ points. The points are partitioned into two sets of equal size, hence each set consists of $4$ points. The segments are stacked on top of each other. The points of a set are also stacked on top of each other. The stack of segments and the stacks of points are arranged in the same way as in the hardness proof for bends.
Each clause gadget is fit into an arc of $1/(4m)$ of a unit circle. The clause gadgets are placed on an arc of one quarter of a unit circle. We call this the \emph{clause arc}. The clause arc is placed next to the variable arc, see Figure~\ref{fig:reduction}. 

\begin{figure}[h!]
\centering
\includegraphics[scale=1.1]{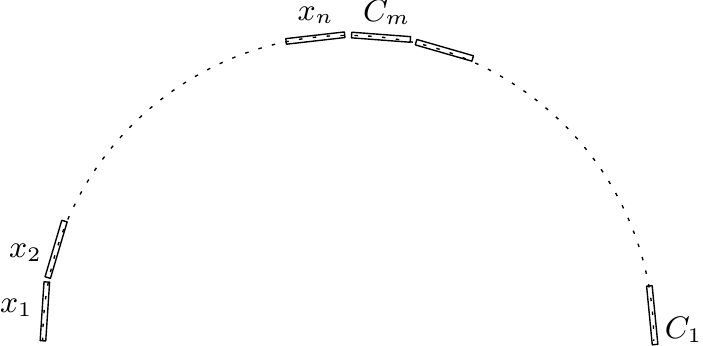}
\caption{Placement of the variable and clause gadgets.}\label{fig:reduction}
\end{figure}

\paragraph*{Connector segments.}
We now place $3m$ connector segments, connecting a variable gadget to a clause gadget whenever the variable appears in the clause. 
Each connector segment consists of a simple line segment. 
A connector segment connects a variable gadget to a clause gadget if the variable appears in the clause. 
The placement of the endpoints of the connector segments within the variable gadgets is as follows: Suppose the variable appears unnegated in the clause. Then the connector segment touches the True path of the gadget and it does not intersect the False path. If the variable appears negated in the clause the segment touches the False path and not the True path. 
The endpoints of the connector segments within the clause gadgets are a bit more involved. We have to ensure that a convex stabber can stab any two of these connector segments in each clause gadget, but not all three.  Figure~\ref{fig:example} shows the placement of the segments endpoints and the three possible ways to traverse the gadget. It can be easily checked that any two of these segments are stabbed by one of the three possible ways. All other possibilities to traverse the gadget stab less than two of the connector segments. Hence there is no way to stab all three connector segments within the clause gadget.

\begin{figure}[h!]
\centering
\includegraphics[scale=0.85]{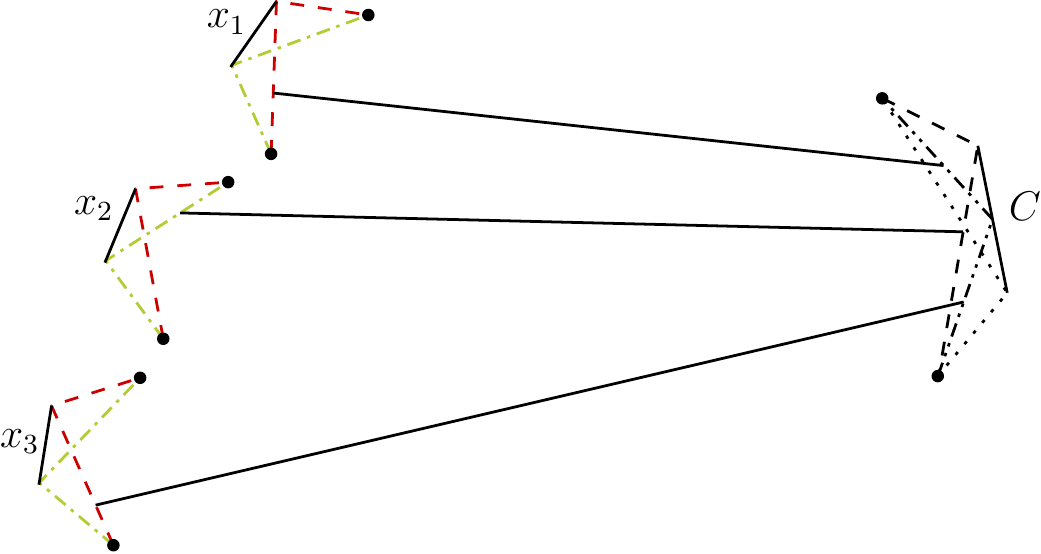}
\caption{Example for a clause $C=\overline{x_1}\vee x_2 \vee \overline{x_3}$.}\label{fig:example}
\end{figure}

Observe that there exists  a convex stabber that stabs all segments of the variable and clause gadgets and at least 2 out of the 3 connector segments of each clause gadget.  Thus, there always exists a stabber that stabs at least $24n+14m$ segments.

\begin{lemma}
There is a convex stabber stabbing $24n+14m+k$ segments if and only if there is an assignment satisfying $k$ clauses.
\end{lemma}
\begin{proof}
If there is an assignment that satisfies $k$ clauses, the stabber traverses the variable gadgets according to the assignment. In each satisfied clause at least one of the connectors is already stabbed, hence the stabber stabs the other two connector segments. 
Thus, it stabs $3$ connector segments for each of the $k$ satisfied clauses and $2$ connector segments for each of the $m-k$ not satisfied clauses. In total the stabber stabs $24n+12m+3k+2(m-k)=24n+14m+k$ segments.
\\
If there is a stabber stabbing $24n+14m+k$ segments, the stabber is forced to traverse the gadgets in \emph{the right order} meaning it either takes the True or False path at each variable gadget and stabs at least two connector segments at each clause gadget. Any stabber traversing the gadgets in any other way stabs less than $24n+14m$ segments. This is because a variable gadget consists of $24$ segments ($18$ of them are segments of length 0), additionally at most $5$ connecting segments have their endpoints inside any variable gadget because each variable appears in at most $5$ clauses. Without traversing the gadget in the right order, a stabber can stab at most $23$ of these segments. A stabber traversing the gadget in the right order stabs at least $24$ segments. For a clause gadget it can be shown that any stabber stabbing the gadget in the right order stabs at least $12$ segments, whereas any other stabber stabs at most $11$ segments.
Thus, a stabber stabs already $24n+14m$ segments by traversing the gadgets in the right order and the $k$ additional segments have to be stabbed in the variable gadgets. Hence, we set the variables according to the stabber traversing these gadgets and so the formula has $k$ satisfied clauses.
\end{proof}

\begin{theorem}
Let $\mathcal L$ be a set of line segments in the plane. It is APX-hard to compute a convex stabber that stabs the maximum number of segments of $\mathcal L$.
\end{theorem}
\begin{proof}
We use a PTAS-reduction from MAX-E3SAT(5).  Let $n$ be the number of variables and $m$ be the number of clauses. We reduce the problem to the convex stabber problem as explained before. Let $k$ be the maximum number of satisfied clauses. Since there always exists an assignment satisfying at least $7/8$ of the clauses it holds $k\geq 7/8m$. Assume there exists a polynomial-time algorithm for maximum convex stabber problem that returns a solution that is at least $(1+\epsilon)$ times of the optimal solution. Then we can approximate MAX-E3SAT(5) by subtracting  $24n+14m$ (note that $3m=5n$):
\begin{eqnarray*}(1+\epsilon)(24n+14m+k)-24n-14m&=&k+\epsilon k +142/5 m\epsilon\leq k+\epsilon k+1136/35\epsilon k\\
&\leq&(1+\epsilon')k 
\end{eqnarray*}
\end{proof}

This construction can be generalized to similar copies of a given convex polygon in the same way as the NP-hardness proof for finding a convex stabber line segments can be generalized to similar copies of a given convex polygon in \cite{DBLP:conf/wads/ArkinDKMPSY11}. So we get the following result.
\begin{theorem}
Let $\mathcal P$ be a set of scaled copies of a convex polygon in the plane. It is APX-hard to compute a convex stabber that stabs the maximum number of segments of $\mathcal P$.
\end{theorem}

Since there exists a PTAS for \emph{planar} MAX-SAT \cite{DBLP:conf/stoc/KhannaM96}, we cannot use these ideas to show APX-hardness for a set of disjoint bends.

\subsection*{Acknowledgment}
I would like to thank Esther M. Arkin and Joseph S. B. Mitchell for pointing out the problem about disjoint objects.
I also want to thank Wolfgang Mulzer for mentioning the optimization problem.

{\small \bibliographystyle{abbrv} \bibliography{paper}}

\begin{thebibliography}{1}

\bibitem{DBLP:conf/wads/ArkinDKMPSY11}
E.~M. Arkin, C.~Dieckmann, C.~Knauer, J.~S.~B. Mitchell, V.~Polishchuk,
  L.~Schlipf, and S.~Yang.
\newblock Convex transversals.
\newblock In F.~Dehne, J.~Iacono, and J.-R. Sack, editors, {\em WADS}, volume
  6844 of {\em Lecture Notes in Computer Science}, pages 49--60. Springer,
  2011.

\bibitem{DBLP:conf/cocoon/BergK10}
M.~de~Berg and A.~Khosravi.
\newblock Optimal binary space partitions in the plane.
\newblock In M.~T. Thai and S.~Sahni, editors, {\em COCOON}, volume 6196 of
  {\em Lecture Notes in Computer Science}, pages 216--225. Springer, 2010.

\bibitem{DBLP:journals/jacm/Feige98}
U.~Feige.
\newblock A threshold of ln {\it n} for approximating set cover.
\newblock {\em J. ACM}, 45(4):634--652, 1998.

\bibitem{DBLP:conf/stoc/KhannaM96}
S.~Khanna and R.~Motwani.
\newblock Towards a syntactic characterization of {PTAS}.
\newblock In G.~L. Miller, editor, {\em STOC}, pages 329--337. ACM, 1996.

\bibitem{tamir-87}
A.~Tamir.
\newblock Problem 4-2 (New York University, Dept. of Statistics and Operations
  Research), Problems Presented at the Fourth NYU Computational Geometry Day
  (3/13/87).

\end{thebibliography}

\end{document}